\documentclass[twocolumn]{autart}

\usepackage{graphicx}      
\usepackage{natbib}        
\usepackage{epsfig} 
\usepackage{color}
\usepackage{mathdots,mathrsfs,bbold}
\usepackage{amssymb,latexsym,amsmath,mathtools,comment}
\usepackage{stmaryrd}
\usepackage[dvips]{psfrag}
\usepackage{url}
\usepackage{setspace}
\mathtoolsset{showonlyrefs=true}

\newtheorem{proof}{Proof}
\newtheorem{theorem}{Theorem}
\newtheorem{lemma}{Lemma}

\newtheorem{proposition}{Proposition}

\newtheorem{assumption}{Assumption}
\newtheorem{definition}{Definition}
 \newcommand{\R}{{\mathbb{R}}}

\newcommand{\Z}{{\mathbb Z}}

\usepackage{pgf,tikz}
\usetikzlibrary{arrows.meta}

\begin{document}
\begin{frontmatter}
\title{Competitive Contagion with Sparse Seeding} 
\author[MIT,NEU]{Milad Siami} 
\author[MIT]{Amir Ajorlou}  and
\author[MIT]{Ali Jadbabaie}  
\address[MIT]{Institute for Data, Systems, and Society, Massachusetts Institute of Technology, Cambridge, MA 02139 USA.  \\ (e-mails: {\{siami, ajorlou, jadbabai\}@mit.edu})}
\address[NEU]{Electrical \& Computer Engineering Department, Northeastern University, Boston, MA 02115 USA. \\ (e-mail: m.siami@northeastern.edu)}
\thanks{This research was supported in part by a Vannevar Bush Fellowship from the Office of Secretary of Defense, DARPA Lagrange, and ARO MURI W911NF-12-1-0509. }

\begin{abstract}           
This paper studies a strategic model of marketing and product diffusion in social networks. We consider two firms
offering substitutable products which can improve their market share by seeding the key individuals in the market. Consumers update their consumption level for each of the two products as the best response to the consumption of their neighbors in the previous period. This results in linear update dynamics for the product consumption. Each consumer receives externality from the consumption of each neighbor where the strength of the externality is higher for consumption of the products of the same firm.  
We represent the above setting as a duopoly game between the firms and introduce a novel framework that allows for sparse seeding to emerge as an equilibrium strategy.
We then study the effect of the network structure on the optimal seeding strategies and the extent to which the strategies can be sparsified.  In particular, we derive conditions under which near Nash equilibrium strategies can asymptotically lead to sparse seeding in large populations. The results are illustrated using a core-periphery network.
\end{abstract}


\end{frontmatter}

\section{Introduction}
\allowdisplaybreaks

Over the past few years, the problem of influence and spread in networks has been subject
to intense study \citep{ballester2006s,bharathi2007competitive,galeotti2009influencing,kempe2003maximizing,kempe2005influential,chasparis2010control,vetta2002nash}. 
Furthermore, modeling and analysis of the spread of new  strategies and behaviors via local coordination games has been an ongoing field of research
\citep{ellison1993learning,kandori1993learning,harsanyi1988general,young1993evolution,young2001individual,young2002diffusion,montanari2010spread,kleinberg2007cascading}. For example, in \cite{lopez2006contagion} the authors show that the contagion of an action in a random network depends on the distribution of the connectivities. In \cite{amini2009marketing}, authors provide an upper bound on the proportion of agents adopting a new product assuming a threshold model for product adoption.  \cite{montanari2010spread} studies the diffusion of innovation in social networks based on the dynamics of coordination games and shows that innovation spreads much more slowly on well-connected network structures dominated by long-range links than in low-dimensional ones dominated by geographic proximity, contrasting  some earlier works on epidemic models (e.g., \citep{ganesh2005effect,draief2006thresholds}).

A game theoretic model of competition and product adoption has been proposed in \cite{goyal2012competitive}. The authors use the proposed model to come up with upper bounds on the price of anarchy and show how network structure may increase the gap between the initial budgets. Similarly, in \cite{bimpikiscompeting}, the authors propose a game theoretic model for competition between firms, where firms can target their marketing budgets toward attracting individual consumers embedded in a social network. They subsequently provide conditions under which it is optimal for the firms to asymmetrically target a subset of the individuals. As another relevant work, \citep{chasparis2010control} considers a dynamical model of preferences in a duopoly setting and characterize optimal policies for both finite and infinite time horizons, studying the effect of endogenous network influences as well as network uncertainties.
The equilibria of network games with linear best response dynamics have been completely characterized in \cite{bramoulle2009strategic}. Considering a monopoly setting, optimal  pricing policies are derived  in ~\cite{candogan2011optimal} assuming quadratic utility functions for the agents, which is a common theme in game-theoretic social network analysis as also previously used in \cite{ballester2006s, corbo2007importance} for instance.

Despite the tremendous development made in the past decade (see e.g., \citep{seeman2013adaptive,fazeli2017competitive,bharathi2007competitive,montanari2010spread,chasparis2010control}), influence maximization algorithms are typically devised under idealized assumptions regarding accessibility of individual consumers for targeted advertisement (e.g., via seeding) as well as the linearity of the cost function with respect to the size of the influence on individual's consumption behavior. In practice, however, firms often have direct access only to a small subset of consumers. In addition, the extent to which firms can influence consumption behavior of consumers is limited, no matter how much they spend on seeding/advertisement. This can be more formally stated as diminishing returns on changes in individual consumption levels, which we model by assuming a convex seeding cost function in our work.
It seems also quite compelling to investigate approaches that can support sparse seeding as (at least) near-optimal marketing strategies, in order to account for the limitations in directly accessing individuals especially in large populations.\footnote{By sparse seeding we mean only seeding a subset of consumers with an infinitesimal size compared to the size of the whole population.}

 In this paper, we study strategic competition between two firms seeking to maximize their product consumption in a network. The consumption of each product by each agent is the result of her myopic best response to the previous consumptions of her peers. A firm can thus improve its market share by targeting its advertising budget toward seeding key individuals in the network, whose consumption of the product can in turn incentivize their peers to consume more of the same product, subsequently affecting the consumption behavior of the individuals all over the network via inter-agent influences.  
 We model the problem above as a fixed-sum game between the two firms, where each firm tries to maximize a utility function which is a discounted sum of its product consumption over time less a seeding cost which is appropriately chosen to account for the diminishing return in seeding/advertisement budget as well. We characterize the unique Nash equilibrium of the resulting duopoly game in terms of the network structure and the market price. The resultant seeding strategy typically prescribes seeding all the agents (in an amount proportional to their influence), which is rather an infeasible task. As a remedy, we propose studying the $\varepsilon$-equilibria of the game and derive conditions under which such equilibria can asymptotically lead to sparse seeding strategies in large populations.
 
 The rest of this paper is organized as follows: in Section~\ref{sec::notations}, we present basic mathematical notations. In Section~\ref{sec::model}, we introduce our model and update dynamics for agents by applying the myopic best response. We then study the game played between the firms and how they decide to seed key individuals in Section~\ref{sec::duopolygame}. Next, we define a near-Nash equilibrium concept as a relaxation of the standard Nash equilibrium with the aim of expanding the equilibrium set to include sparse seeding strategies in Section~\ref{sec::sparse}. In Section~\ref{sec::asymp}, we consider the case of a large population and characterize network structures for which a pair of sparse seeding strategies can be asymptotically realized. The results are illustrated via an example  in Section~\ref{sec::discussion}. Finally, in Section~\ref{sec::conclusions}, we conclude the paper.
 
\section{Mathematical Notations}
\label{sec::notations}
\allowdisplaybreaks

	Throughout the paper,
the discrete time index is denoted by $k$. The sets of real (integer), positive real (integer), and strictly positive real (integer) numbers are represented by $\R$ ($\mathbb Z$), $\R_+$ ($\mathbb Z_+$) and $\R_{++}$ ($\mathbb Z_{++}$), respectively. The set of natural numbers $\{i \in \Z_{++} ~:~i \leq n\}$ is denoted by $[n]$. 
Bold letters, such as $\mathbf x$ or $\mathbf s$, stand for real-valued vectors. Capital letters, such as $A$ or $B$, stand for real-valued matrices. 
We use $\|\mathbf x\|_2$ to denote the $\ell^2$-norm of vector $\mathbf x$.
We denote the number of nonzero elements in vector $\mathbf x$ by $\|\mathbf x\|_0$. 
The $n$-by-$n$ identity matrix is denoted by $I_n$. Also, we represent the $n$-by-$1$ vector of ones by $\mathbf 1_n$ and the $n$-by-$1$ and $n$-by-$n$ matrices of zeros by $\mathbf 0_n$ and $\mathbf 0_{n\times n}$, respectively.
 The transpose of matrix $A$ is denoted by $A^\top$. 
For two matrices $A$ and $B$, 
we denote the Hadamard product by $A \circ B$ and the Kronecker product by 
$A\otimes B$.

\section{Spread Dynamics}
\label{sec::model}
We consider a social network consisting of a group of $n$ consumers (agents) denoted by $\mathcal V =[n]$. The relationship among agents is given by a weighted directed graph $\mathcal G =(\mathcal V, \mathcal E, w)$. The weighted adjacency matrix of $\mathcal G$ is denoted by $G$ where its $i,j$-th entry denoted by $g_{ij} = w((i,j))$ if $(i,j)\in \mathcal E$ otherwise $g_{ij}=0$. Link weight $g_{ij}$ presents the strength of the influence of agent $j$ on $i$. 

We assume there are two competing firms producing product $a$ and $b$, respectively. Let $d_i^{\text{in}} =\sum_j g_{ij}$ and $d_i^{\text{out}} =\sum_j g_{ji}$ denote the in-degree and out-degree of node $i$, respectively.
{Assume $\bar x_i(k)$ and $\underline{x}_i(k)$ denote agent $i$'s consumption of product $a$ and $b$ at time $k\in\Z_{++}$, respectively.
The initial consumption of agent $i$ from each product is determined by the effort 
made by each firm in seeding/marketing its product to agent $i$ at time $0$, and are denoted by
$\bar{s}_i$ and $\underline{s}_i$. We refer to $\bar s_i$ and $\underline s_i$ as the control/seeding of firm $a$ and $b$ on agent $i$, respectively.  The relation between the controls and the states is motivated by \citep{chasparis2010control} on social networks. We consider a convex cost function $c(\mathbf s)$ for seeding to reflect the diminishing return of seeding/advertisement budget on changing the consumption behavior in the population. The cost $c(\mathbf s)$ reflects the monetary value required to increase the consumption of the agents in the population by amount $\mathbf s$. For the sake of simplicity, we develop our results for a quadratic cost function of the form $c(\mathbf s)=\frac{1}{2}\|\mathbf s\|_2^2$.}

The total utility of agent $i$ from taking action $x_i$ is given by
\begin{align}
{\mathfrak{u}}_i(\bar x_i(k), \bar{\mathbf x}_{-i}(k)) &= \alpha \bar x_i(k) - \frac{1}{2} \left(\bar x_i(k)\right)^2\\
& + ~\bar x_i(k) \sum_{j \sim i} g_{ij}\left( \bar x_j(k) + \beta {\underline x}_j(k) \right)\\
& - ~p \,\bar x_i(k)
\end{align}
where $0 \leq \beta < 1$. In the above equation $\bar{\mathbf x}_{-i}$ denotes an action vector of all agents but agent $i$. The first part of the utility ($\alpha \bar x_i(k) - \frac{1}{2} \left(\bar x_i(k)\right)^2$) represents the normalized second order approximation of a concave self-utility function, as is commonly used in the literature. The fact that externality is weaker for consuming different products is captured by $\beta<1$. Finally, the price $p$ represents the common market price for the products.

\subsection{Myopic Best Response Dynamics}
We assume agents repeatedly apply myopic best response to the consumption levels of their neighbors in the previous stage to update their consumption of each product, that is,
\[ \bar{x}_i(k+1)~\in~\arg \max_{x \in \R_+^n} \mathfrak{u}_i \left (x, \bar{\mathbf x}_{-i}(k)\right),\]
and
\[ \underline{x}_i(k+1)~\in~\arg \max_{x \in \R_+^n} \mathfrak{u}_i\left(x, \underline{\mathbf x}_{-i}(k)\right).\]

Consumption levels for agent $i$ at time $k + 1$ thus find the following linear dynamics:
\[\bar x_i(k+1) = \left (\alpha - p\right) + \sum_{j \sim i} g_{ij}\left( \bar x_j(k) + \beta \underline x_j(k) \right). \]
and
\[\underline {x}_i(k+1) = \left (\alpha - p\right) + \sum_{j \sim i} g_{ij}\left( \underline{x}_j(k) + \beta \bar x_j(k) \right). \]
This results in the following closed-form update dynamics:
\begin{equation}
\begin{cases}
\bar{\mathbf x}(k+1)~=~(\alpha-p) \mathbf 1_n \,+\, G\bar{\mathbf x}(k)\,+\,\beta G\underline{\mathbf x}(k)\\
\underline{\mathbf x}(k+1)~=~(\alpha-p) \mathbf 1_n \,+\, G \underline{\mathbf x}(k)\,+\,\beta G \bar {\mathbf x}(k)
\end{cases}
\label{eq::1}
\end{equation}
where $\bar {\mathbf x}(0) = \bar {\mathbf s}$ and $\underline{\mathbf x}(0) = \underline {\mathbf s}$.

\begin{assumption}
\label{assum1}
We assume $\alpha\geq p$ to guarantee that $\bar x_i(k)$ and $\underline x_i(k)$ are non-negative for all feasible initial seedings.
\end{assumption}


In the next section, we show how firms can exploit the structure of the network to maximize their product consumption, and we then characterize the unique Nash equilibrium of the game played between these two firms.

We first recall the definition of the Katz-Bonacich centrality measure.
\begin{definition}
\label{def::centrality}
For a given attenuation factor $\alpha$ less than the reciprocal of the absolute value of the largest eigenvalue of the adjacency matrix $G \in \R_+^{n \times n}$, the Katz-Bonacich  centrality is given by
\[ {\boldsymbol {{\mathfrak c}_{\text{katz}}}}(G, \alpha)~=~\left(I_n-\alpha G^\top\right)^{-1}\mathbf 1_n.\]
\end{definition}

\begin{definition}
\label{def::one}
For a given subset $\mathcal S \subseteq \mathcal V$, we define an indicator binary $n \times 1$ vector $\mathbf 1_{\mathcal S}$ as follows:
\begin{equation}
\mathbf 1_{\mathcal S}=[s_i]_{n \times 1}, ~~~s_i =\begin{cases} 1~~\text{if}~i \in \mathcal S\\
0,~~i \in [n] \setminus \mathcal S\end{cases}
\end{equation}
\end{definition}

We then use these definitions in the next section to characterize  equilibrium strategies in a duopoly game between the firms.

\section{Optimal Seeding Strategies}\label{sec::duopolygame}

This section describes the game between two firms where each firm aims to maximize the consumption of its product over an {infinite} horizon.
Each firm can invest in promoting its product by seeding some of the agents.
Initial seeding could be viewed as free offers to promote the product in the social networks. We thus consider the problem of deriving optimal advertising policies for the spread of innovations/consumption in a network. In \cite{chasparis2010control}, an analytical solution to the optimal advertising problem in the absence of a competing firm is provided, and it is shown that the solution can be related to previously introduced centrality measures in sociology.

We define the utility of each firm as the discounted sum of its product consumption over time minus the squared norm of its seeding, which are defined formally as follows:
\begin{equation}
    \bar{\mathfrak{U}}~=~p \left( \sum_{k=1}^{\infty} \delta^k \mathbf 1_n ^\top \bar{\mathbf x}(k)\right)\, -\, \frac{1}{2}\|\bar{\mathbf s} \|^2_2,
    \label{eq:ua}
\end{equation}
and
\begin{equation}
    \underline{\mathfrak{U}}~=~p \left( \sum_{k=1}^{\infty} \delta^k \mathbf 1_n ^\top \underline{\mathbf x}(k)\right)\, -\, \frac{1}{2}\|\underline{\mathbf s} \|^2_2.
    \label{eq:ub}
\end{equation}

\begin{definition}
\label{def::nash}
A pair of seeding strategies $(\bar{\mathbf s}^\star, \underline{\mathbf s}^\star) \in \R_+^n \times \R_+^n$  is said to be a Nash equilibrium of the duopoly game described above if  none of the players can improve her payoff by unilaterally deviating from her strategy. That is,
\[\underline{\mathfrak{U}}(\bar{\mathbf s}^\star, \underline{\mathbf s}) ~\leq~ \,\underline{\mathfrak{U}}(\bar{\mathbf s}^\star, \underline{\mathbf s}^\star),~~\forall \,  \underline{\mathbf s} ~ \in ~\R^n_{+},  \]
and
\[\bar{\mathfrak{U}}(\bar{\mathbf s}, \underline{\mathbf s}^\star) ~\leq~\, \bar{\mathfrak{U}}(\bar{\mathbf s}^\star, \underline{\mathbf s}^\star), ~~\forall \,  \bar{\mathbf s} ~ \in ~\R^n_{+}.  \]
\end{definition}

{\begin{assumption}
\label{assum3}
We assume that the absolute value of the largest eigenvalue of $G$ is less than or equal to $\Lambda_{\max}$ where
\[0<\Lambda_{\max} < \frac{1}{\delta(1+\beta)}.\]
\end{assumption}}

We then use Assumption \ref{assum3} to get well-defined centrality measures for the graph with adjacency $G$, and make the matrix pencil in Definition \ref{def::centrality} invertible.

In the following lemma, an equilibrium strategy for the duopoly game between the firms is characterized based on the node centrality measures.

\begin{lemma}
\label{lemma}
Consider two firms with closed-form update dynamics \ref{eq::1}, and  utility functions $\bar {\mathfrak U}$ and $\underline{\mathfrak U}$ given by \eqref{eq:ua} and \eqref{eq:ub}, respectively. The sensitivity of the utilities with respect to the individual seedings are given by
\begin{equation}
\frac{\partial \bar{\mathfrak U}}{\partial \bar s_i} ~=~ p \, c_i - \bar s_i,
\end{equation}
and
\begin{equation}
\frac{\partial \underline{\mathfrak U}}{\partial \underline{s}_i} ~=~ p \, c_i - \underline{s}_i,
\end{equation}
where ${\mathfrak c}_{\text{new}} =[c_1, \ldots, c_n]^\top$ with
\begin{equation}
{\mathfrak c}_{\text{new}} ~=~\frac{1}{2}\boldsymbol{{\mathfrak c}_{\text{katz}}} \left (G, \delta (1-\beta)\right)+\frac{1}{2}\boldsymbol{{\mathfrak c}_{\text{katz}}}\left (G, \delta (1+\beta)\right).
\label{eq:c}
\end{equation}
\end{lemma}

\begin{proof}
Let us define the following Jacobian matrices for firm $a$ as follows:
\begin{equation}
\bar{\boldsymbol{\mathfrak x}}(k) := \frac{\partial (\bar{\mathbf x} (k))}{\partial \bar{\mathbf s}},
\label{eq:308}
\end{equation}
and
\begin{equation}
 \bar {\boldsymbol{\mathfrak X}}(k) := \frac{\partial \sum_{t=1}^{k} \delta^t  \bar{\mathbf x}(t)}{\partial \bar {\mathbf s}}= \sum_{t=1}^{k} \delta^t \bar{\mathfrak x}(t).
\label{eq:292}
\end{equation}
Similarly, we define  the following Jacobian matrices for firm $b$:
\begin{equation}
\underline{\boldsymbol{\mathfrak x}}(k) := \frac{\partial ( \underline{\mathbf x} (k))}{\partial \bar{\mathbf s}},
\label{eq:323}
\end{equation}
and
\begin{equation}
 \underline{\boldsymbol{\mathfrak X}}(k) ~:=~\frac{\partial \sum_{t=1}^{k} \delta^t \underline{\mathbf x}(t)}{\partial \bar{\mathbf s}}= \sum_{t=1}^{k} \delta^t \underline{\mathfrak x}(t).
 \label{eq:303}
 \end{equation}
Next, we write the update dynamics for $\bar{\boldsymbol{\mathfrak x}}$ and $\underline{\boldsymbol{\mathfrak x}}$ according to their definitions and update dynamics \eqref{eq::1}:
\begin{equation}
\begin{cases}
\bar{\boldsymbol{\mathfrak x}}(k+1)~=~\bar{\boldsymbol{\mathfrak x}}(k) \, G^\top +  \underline {\boldsymbol{\mathfrak x}}(k) \,  \beta G^\top \\
\underline{\boldsymbol{\mathfrak x}}(k+1)~=~  \underline{\boldsymbol{\mathfrak x}}(k) \, G^\top+  \bar{\boldsymbol{\mathfrak x}}(k) \, \beta G^\top
\end{cases}
\label{eq:316}
\end{equation}
Then, based on update dynamics \eqref{eq:316}, \eqref{eq:292} and \eqref{eq:303} we get
\begin{equation}
\begin{cases}
\bar{\boldsymbol{\mathfrak X}}(\infty) - \bar {\boldsymbol{\mathfrak x}}(0)~=~\delta  \bar{\boldsymbol{\mathfrak X}}(\infty) \, G^\top +  \delta \underline {\boldsymbol{\mathfrak X}}(\infty) \, \beta G^\top\\
\underline{\boldsymbol{\mathfrak X}}(\infty) - \underline{\boldsymbol{\mathfrak x}}(0)~=~\delta  \underline{\boldsymbol{\mathfrak X}}(\infty) \, G^\top + \delta \bar{\boldsymbol{\mathfrak X}}(\infty) \, \beta  G^\top
\end{cases}
\end{equation}
where
$ \bar {\boldsymbol{\mathfrak{x}}} (0) = I_n$,
and
$\underline {\boldsymbol{\mathfrak{x}}} (0) = \mathbf 0_{n\times n}$.
We can rewrite \eqref{eq:325} in the following compact form
\begin{eqnarray}
&&\begin{bmatrix}
\bar{\boldsymbol{\mathfrak X}}(\infty)&
\underline{\boldsymbol{\mathfrak X}}(\infty)
\end{bmatrix}
\left (I_{2n} - \delta \begin{bmatrix}
G^\top & \beta G^\top\\
\beta G^\top & G^\top
\end{bmatrix}\right)\nonumber \\  &&~~~~~~~~~~~~~~ = \begin{bmatrix}
I_{n}&
\mathbf 0_{n \times n}
\end{bmatrix}.
\label{eq:325}
\end{eqnarray}
From this, it follows that
\begin{small}
\begin{equation}
\bar{\boldsymbol{\mathfrak X}}(\infty)=  \begin{bmatrix}
I_n & \mathbf 0_{n\times n}
\end{bmatrix}\left (I_{2n}- \delta \begin{bmatrix}
G^\top & \beta G^\top\\
\beta G^\top & G^\top
\end{bmatrix}\right)^{-1} \begin{bmatrix}
I_n\\
\mathbf 0_{n \times n}
\end{bmatrix}.
\label{eq:339}
\end{equation}
\end{small}
 We then use the following property to simplify \eqref{eq:339}
 \begin{equation}
({A} \otimes {B})({C} \otimes {D}) ~=~ (A{C}) \otimes (B{D}).
 \label{eq:hadam}
 \end{equation}
Let us define
\begin{equation}
\mathcal A ~:=~ \begin{bmatrix}
G^\top & \beta G^\top\\
\beta G^\top & G^\top
\end{bmatrix} ~=~ \begin{bmatrix}
1 & \beta \\
\beta & 1
\end{bmatrix} \otimes G^\top.
\label{eq:A}
\end{equation}
based on \eqref{eq:hadam} and \eqref{eq:A}, we get
\begin{equation}
\mathcal A^t = \left(\begin{bmatrix}
1 & \beta \\
\beta & 1
\end{bmatrix} \otimes G^\top \right)^t = \begin{bmatrix}
1 & \beta \\
\beta & 1
\end{bmatrix}^t \otimes (G^\top)^t.
\label{eq:At}
\end{equation}
With a simple calculation we get:
\begin{equation}
 \begin{bmatrix}
1 & 0
\end{bmatrix}\begin{bmatrix}
1 & \beta \\
\beta & 1
\end{bmatrix}^t \begin{bmatrix}
1  \\
0
\end{bmatrix}=\frac{1}{2} \left((1-\beta)^t+(1+\beta)^t\right).
\label{eq:411}
\end{equation}
Next by expanding $(I - \delta \mathcal A)^{-1}$ and then applying \eqref{eq:hadam}, \eqref{eq:At}, and \eqref{eq:411}, it follows that
\begin{eqnarray*}
&& \begin{bmatrix}
I_n & \mathbf 0_{n \times n}
\end{bmatrix}\left (I_{2n}- \delta \begin{bmatrix}
G^\top & \beta G^\top\\
\beta G^\top & G^\top
\end{bmatrix}\right)^{-1} \begin{bmatrix}
\mathbf 1_n\\
\mathbf 0_n
\end{bmatrix} \\
&&~~~~~=~\frac{1}{2}\Big(I_n - \delta (1-\beta) G^\top \Big)^{-1}\mathbf 1_{n} \\
&&~~~~~~~~~+~ \frac{1}{2}\Big(I_n - \delta (1+\beta) G^\top \Big)^{-1} \mathbf 1_{n}\\
&&~~~~~=~\frac{1}{2}\boldsymbol{{\mathfrak c}_{\text{katz}}} \left (G, \delta (1-\beta)\right)+\frac{1}{2}\boldsymbol{{\mathfrak c}_{\text{katz}}}\left (G, \delta (1+\beta)\right).
\label{eq:405}
\end{eqnarray*}
By substituting \eqref{eq:339} in \eqref{eq:375}, we have
\begin{equation}
\boldsymbol{\mathfrak X}(\infty) \mathbf 1_n ~=~\frac{1}{2}\boldsymbol{{\mathfrak c}_{\text{katz}}} \left (G, \delta (1-\beta)\right)+\frac{1}{2}\boldsymbol{{\mathfrak c}_{\text{katz}}}\left (G, \delta (1+\beta)\right).
\label{eq:430}
\end{equation}
From \eqref{eq:ua}, we get
\begin{equation}
\frac{\partial \bar{\mathfrak U}}{\partial \bar{\mathbf s}} ~=~ p \, \bar{\boldsymbol{\mathfrak X}}(\infty) \mathbf 1_n - \bar {\mathbf s}.
\label{eq:375}
\end{equation}
Finally, by substituting \eqref{eq:430} in \eqref{eq:375}, we get the desired result.
\end{proof}

\begin{theorem}
\label{theorem::Nash}
Consider two firms with closed-form update dynamics \eqref{eq::1}, and utility functions $\bar {\mathfrak U}$ and $\underline{\mathfrak U}$ given by \eqref{eq:ua} and \eqref{eq:ub}, respectively. The game between firms admits the unique symmetric  Nash equilibrium of the form
\begin{equation}
\bar{\mathbf s}^\star=\underline{\mathbf s}^\star= {p} {\mathfrak c}_{\text{new}}, 
\end{equation}
where bi-product centrality vector ${\mathfrak c}_{\text{new}}$ is given by \eqref{eq:c}.
\end{theorem}

\begin{proof}
The proof is a direct consequence of Lemma \ref{lemma}.
\end{proof}

\section{Sparse Seeding}
\label{sec::sparse}

In what follows, we first define a near-Nash equilibrium concept as a relaxation of the standard Nash equilibrium defined in Definition~\ref{def::nash}, with the aim of expanding the equilibrium set to include sparse seeding strategies (cf. \citep{daskalakis2006note}). In what follows, we will make this statement
formal.

\begin{definition}
\label{def::varnash}
Given $\varepsilon \in \R_{+}$, a pair of seeding strategies $(\bar{\mathbf s}^\star, \underline{\mathbf s}^\star) \in \R_+^n \times \R_+^n$  is said to be $\varepsilon$-equilibrium of the duopoly game described in Section~\ref{sec::duopolygame} if  none of the players can improve her payoff by an amount more than $\varepsilon$ fraction of her current payoff,
by unilaterally deviating from her strategy. That is,
\[\underline{\mathfrak{U}}(\bar{\mathbf s}^\star, \underline{\mathbf s}) ~\leq~ (1+\varepsilon)\, \underline{\mathfrak{U}}(\bar{\mathbf s}^\star, \underline{\mathbf s}^\star), ~~\forall \,  \underline{\mathbf s} ~ \in ~\R^n_{+}, \]
and
\[\bar{\mathfrak{U}}(\bar{\mathbf s}, \underline{\mathbf s}^\star) ~\leq~  (1+\varepsilon)\, \bar{\mathfrak{U}}(\bar{\mathbf s}^\star, \underline{\mathbf s}^\star), ~~\forall \,  \bar{\mathbf s} ~ \in ~\R^n_{+}. \]
\end{definition}

Every Nash Equilibrium is equivalent to a $\varepsilon$-equilibrium where $\varepsilon =0$.

The next lemma provides a closed-form expression for the utility functions of the firms when there is not any seeding.

\begin{lemma}
\label{lemma:2}
Assume $\bar{\mathbf s}=\underline{\mathbf s}=\mathbf 0_n$, then the utility of firm $a$ is reduced to
\[ {\bar {\mathfrak U}(\mathbf 0_n,\mathbf 0_n)  ~=~\frac{(\alpha p-p^2)\delta}{1-\delta} \, \mathbf 1_n^{\top}\boldsymbol{{\mathfrak c}_{\text{katz}}}(G,\delta(1+\beta))},\]
where centrality $\boldsymbol{{\mathfrak c}_{\text{katz}}}(G,\cdot)$ is given by Definition \ref{def::centrality}.
\end{lemma}

\begin{proof}
Let us assume $\bar{\mathbf s}=\underline{\mathbf s}=\mathbf 0_n$, then
\begin{equation}
\sum_{k=1}^{\infty}\delta^k\begin{bmatrix}
\bar{\mathbf x}(k)\\  \underline{\mathbf x}(k)
\end{bmatrix}~=~\frac{\delta(\alpha-p)}{1-\delta} \left (I_{2n} - \delta \mathcal A \right)^{-1}\begin{bmatrix}
\mathbf 1_n\\  \mathbf 1_n
\end{bmatrix},
\end{equation}
where $\mathcal A$ is given by \eqref{eq:A}. 
With a simple calculation we get:
\begin{equation}
 \begin{bmatrix}
1 & 0
\end{bmatrix}\begin{bmatrix}
1 & \beta \\
\beta & 1
\end{bmatrix}^k \begin{bmatrix}
1  \\
1
\end{bmatrix}~=~(1+\beta)^k.
\end{equation}
Then, based on the definition of the utility, we have
\begin{eqnarray*}
\bar {\mathfrak U}(\mathbf 0_n,\mathbf 0_n) &=&\frac{\delta(\alpha-p)}{1-\delta} \begin{bmatrix}
\mathbf 1&  \mathbf 0
\end{bmatrix} \left (I_{2n} - \delta \mathcal A \right)^{-1}\begin{bmatrix}
\mathbf 1\\  \mathbf 1
\end{bmatrix}\\
&=& p\frac{\delta(\alpha-p)}{1-\delta} \mathbf 1_n^{\top}\boldsymbol{{\mathfrak c}_{\text{katz}}}(G,\delta(1+\beta)).
\end{eqnarray*}
\end{proof}

{
\begin{theorem}Consider two firms with closed-form update dynamics \eqref{eq::1}, and utility functions $\bar {\mathfrak U}$ and $\underline{\mathfrak U}$ given by \eqref{eq:ua} and \eqref{eq:ub}, respectively.  For any given sets $\bar{\mathcal S}, \underline{\mathcal S} \subseteq [n]$, the game between the firms admits a $\varepsilon$-equilibrium of the form
\begin{equation}
\bar{\mathbf s}^\star~=~ p \, ({\mathfrak c}_{\text{new}}\, \circ \,  \mathbf 1_{\bar{\mathcal S}}),~~{\text{and}}~~\underline{\mathbf s}^\star~=~ p \, ({\mathfrak c}_{\text{new}}\, \circ \,  \mathbf 1_{\underline {\mathcal S}}), \label{eq::enash}
\end{equation}
where $\mathbf c_{\text{new}}$ is given by \eqref{eq:c}, if and only if
\begin{equation}
  \varepsilon ~\geq~ \max (\bar \tau, \underline \tau ),
   \label{eq::var}
\end{equation}
where
\[\bar \tau~:=~\frac{  \sum_{i \notin \bar{ \mathcal S}} c_i^2 }{\frac{\delta(\alpha-p)}{2p(1-\delta)} \mathbf 1_n^{\top}\boldsymbol{{\mathfrak c}_{\text{katz}}}(G,\delta(1+\beta))+ \sum_{i \in \bar {\mathcal S}} c_i^2},\]
\[\underline \tau~:=~\frac{  \sum_{i \notin \underline {\mathcal S}} c_i^2 }{\frac{\delta(\alpha-p)}{2p(1-\delta)} \mathbf 1_n^{\top}\boldsymbol{{\mathfrak c}_{\text{katz}}}(G,\delta(1+\beta))+ \sum_{i \in \underline {\mathcal S}} c_i^2},\]
and centrality $\boldsymbol{{\mathfrak c}_{\text{katz}}}$ is given in Definition \ref{def::centrality}.
\end{theorem}}

\begin{proof}
We first show that
\begin{equation}
    \max_{\bar{\mathbf s} \in \R_+^n} \left( \bar{\mathfrak{U}}(\bar{\mathbf s}, \underline{\mathbf s}^\star) - \bar{\mathfrak{U}}(\bar{\mathbf s}^\star, \underline{\mathbf s}^\star)\right) = \frac{1}{2}p^2\,\sum_{i \notin \mathcal S} c_i^2,
    \label{eq::eps}
    \end{equation}
    which results in the following best response for firm $a$
    \begin{equation}
    \bar{s_i} ~=~ p c_i,
    \label{eq:920}
    \end{equation}
where $i \in [n]$.
    Therefore, using \eqref{eq:920}, the superposition property, and Lemma \ref{lemma:2}, it follows that
    \begin{align}
    &\max_{\bar{\mathbf s} \in \R_+^n} \, \bar{\mathfrak{U}}(\bar{\mathbf s}, \underline{\mathbf s}^\star) ~=~ \bar {\mathfrak U}(\mathbf 0_n,\mathbf 0_n) + \frac{1}{2}p^2\,\sum_{i \notin \mathcal S} c_i^2 \\
    &~~~=\frac{(\alpha p-p^2)\delta}{1-\delta} \, \mathbf 1_n^{\top} \boldsymbol{{\mathfrak c}_{\text{katz}}}(G,\delta(1+\beta))+\frac{1}{2}p^2\,\sum_{i \notin \mathcal S} c_i^2.
    \label{eq:923}
    \end{align}
    Then using \eqref{eq::eps}, \eqref{eq:923} and Definition \ref{def::varnash}, we get the desired result.
\end{proof}




\section{Asymptotically Realizable Sparse Equilibria}
\label{sec::asymp}

In this section, we consider the case of a large population for which $n\to\infty$. We characterize network structures for which a pair of sparse seeding strategies can be realized as the limit of a sequence of $\varepsilon$-equilibria with $\epsilon\to0$. In what follows, we will make this statement
formal.

\begin{definition}
We call a pair of $\varepsilon$-equilibrium seeding strategies $(\bar {\mathbf s}^\star,\underline{\mathbf s}^\star)$ asymptotically sparse-realizable (ASR) if and only if $\|\bar{\mathbf s}^\star\|_0 = \mathcal O(1)$, $\|\underline{\mathbf s}^\star\|_0 = \mathcal O(1)$ and $\varepsilon = o(1)$.\footnote{Given functions $f(\cdot)$ and $g(\cdot)$, the asymptotic notations $f(n) = \mathcal O(g(n))$ and $f(n) = o(g(n))$ mean $\limsup_{n\to \infty} \left |\frac{f (n)}{g(n)} \right | < \infty$ and $\lim_{n\to \infty} \left |\frac{f (n)}{g(n)} \right | = 0$, respectively.}
\end{definition}

We begin our analysis by a lemma that presents a necessary and sufficient condition for a pair of strategies be asymptotically sparse-realizable in terms of the bi-product centrality \eqref{eq:c}.

{\begin{lemma}
\label{lemma:3}
A pair of strategies $(\bar {\mathbf s}^\star,\underline{\mathbf s}^\star)$  is asymptotically sparse-realizable if and only if 
\begin{equation}
\max \left \{ \frac{\sum_{i\notin \bar{\mathcal S}}c_i^2}{\sum_{i \in [n] }c_i^2}, \frac{\sum_{i\notin \underline {\mathcal S}}c_i^2}{\sum_{i \in [n] }c_i^2} \right \} ~=~o(1), 
\label{eq:iff}
\end{equation}
where $\mathbf c=[c_1, \cdots, c_n]^\top$ is given by \eqref{eq:c}, $\bar{\mathcal S} = \{i ~|~ \bar{s}_i^\star \neq 0\}$, $|\bar{\mathcal S}| = \mathcal O(1)$, $\underline{\mathcal S} = \{i ~|~ \bar{s}_i^\star \neq 0\}$, and $|\underline{\mathcal S}| = \mathcal O(1)$.
\end{lemma}}

\begin{proof}
We first start with the fact that $\varepsilon = o(1)$ if and only if $\frac{\varepsilon}{\varepsilon+1} = o(1)$. Therefore, the pair of $\varepsilon$-equilibrium seeding strategies is asymptotically sparse-realizable if and only if 
\begin{small}
\begin{equation}
  \frac{  \sum_{i \notin \mathcal S} c_i^2 }{\frac{\delta(\alpha-p)}{2p(1-\delta)} \mathbf 1_n^{\top}\boldsymbol{{\mathfrak c}_{\text{katz}}}(G,\delta(1+\beta))+ \sum_{i \in [n]} c_i^2} ~=~ o(1).
\end{equation}
\end{small}
Let us recall
   \begin{align}
\mathbf c &=[c_1, \cdots, c_n]^\top\\
&:= \frac{1}{2}\boldsymbol{{\mathfrak c}_{\text{katz}}} \left (G, \delta (1-\beta)\right)+\frac{1}{2}\boldsymbol{{\mathfrak c}_{\text{katz}}} \left (G, \delta (1+\beta)\right),
   \end{align}
where $\mathbf a~=~[a_1, \cdots, a_n]^\top:=~\boldsymbol{{\mathfrak c}_{\text{katz}}} \left (G, \delta (1-\beta)\right)$,
and $\mathbf b~=~[b_1, \cdots, b_n]^\top~:=~\boldsymbol{{\mathfrak c}_{\text{katz}}} \left (G, \delta (1+\beta)\right)$.
It can be seen that 
\[  \sum_{i \in [n]} a_i ~\leq~ 2 \sum_{i \in [n]} c_i ~\leq~ 2 \sum_{i \in [n]} c_i^2,\]
where in the last inequality, we use the fact that $c_i \geq 1$.
{
\begin{small}
\begin{equation}
\frac{\sum_{i\notin S}c_i^2}{\sum_{i \in [n] }c_i^2} \geq \frac{\sum_{i\notin S}c_i^2}{\sum_{i \in [n] }c_i^2+ \kappa \sum_{i \in [n]} a_i} \geq \frac{\sum_{i\notin S}c_i^2}{(1+2 \kappa)\sum_{i \in [n] }c_i^2} 
\end{equation}
\end{small}}
where $\kappa = \frac{\delta(\alpha-p)}{2p(1-\delta)} \geq 0$.

\end{proof}

As the next result, we derive a necessary condition for the existence of asymptotically sparse-realizable  equilibrium strategies. 

{\begin{proposition}
Let $d^{\text{out}}_{\max} = {\max_{i \in [n]}} \, d_i^{\text{out}}$, and suppose that 
\[\delta(1+\beta) \, d^{\text{out}}_{\max}<1.\]
 Then, there exists no pair of $\varepsilon$-equilibrium seeding strategies that is asymptotically sparse-realizable. 
\end{proposition}}

\begin{proof}
It can be seen that
\begin{small}
\[1+\delta\beta d_i^{\text{out}}~\leq~ c_i~\leq~ \frac{1-\delta\beta d^{\text{out}}_{\max}}{(1-\delta(1+\beta)d^{\text{out}}_{\max})(1-\delta(1-\beta)d^{\text{out}}_{\max})},\]
\end{small}
where $\mathbf c=[c_1, \cdots, c_n]$ is given by \eqref{eq:c}.
Therefore, using these inequalities, it can be seen that if $\delta(1+\beta) d_{max}^{\text{out}} < 1$ then $c_i$'s are bounded; therefore \eqref{eq:iff} does not hold. This completes the proof.

\end{proof}

This result implies that networks with bounded out-degree (i.e.,  $d_{\max}^{\text{out}} = \mathcal O(1)$) are not asymptotically sparse-realizable. 

We conclude this section by stating a sufficient condition for the existence of asymptotically sparse-realizable equilibrium strategies.
\begin{proposition}
Suppose that the pair of seeding strategies \eqref{eq::enash} is asymptotically sparse-realizable. Then, \[\max_{i \in [n]} c_i ~=~ \mathcal O(n),\] where $\mathbf c=[c_1, \cdots, c_n]^\top$ is given by \eqref{eq:c}. 
\end{proposition}
\begin{proof}
The proof is a direct consequence of Lemma \ref{lemma:3}.

\end{proof}

\section{Discussion}
\label{sec::discussion}
{We illustrate our results using a core-periphery network structure, that is a network with few highly interconnected and many sparsely connected nodes.
\begin{figure}
\centering
\include{ex3}
\caption{A core-periphery network consisting of three communities each with one role model that influences every community member by $g$. Each role model herself is influenced by the role model of another community by amount $g$.}
\label{fig::coreperiphery}
\end{figure}

Consider a network consisting of $\chi$ communities of size $m$ denoted by $\mathcal C_r~=~\left \{(r-1)m+1,\cdots,rm \right\}$
for $r\in[\chi]$. Consumption levels of the agents within each community are influenced by an agent  called a \textit{role model} that induces an externality effect of magnitude $g$ on each consumer. Consumption levels of role models themselves are each influenced by the consumption of a role model in another community whose structure is assumed to form a cycle. Let agent $rm$ be the role model in community $\mathcal C_r$ for $r\in [\chi]$. The corresponding entries of the adjacency matrix $G$ are given by
\begin{equation}
g_{ij}\!=\!\begin{cases}
g,&\hspace*{-6pt}\text{for}\,i\in \mathcal C_r\setminus\{rm\}~\text{and}~j=rm\\
g,&\hspace*{-6pt}\text{for}\,(i,j)\!\in\!\{((r\!+\!1)m,rm)|r\!\in\![\chi\!-\!1]\}\!\cup\!\{(m,\chi m)\}\\
0,&\hspace*{-6pt}\text{otherwise}
\end{cases}
\end{equation}
The case $\chi=3$ is depicted in Fig.~\ref{fig::coreperiphery}. 

Characterizing the equilibrium seeding strategies (both Nash equilibrium and $\varepsilon$-equilibria) requires finding centrality vectors \[\mathbf a~=~\left(I_n-\delta(1-\beta) G^\top\right)^{-1}\mathbf 1_n,\] and 
\[\mathbf b~=~\left(I_n-\delta(1+\beta) G^\top\right)^{-1}\mathbf 1_n,\] where $n=\chi m$. Let $a_L$ and $a_F$ be the corresponding centralities of a role model and a periphery consumer in $\mathbf a$, respectively. A periphery consumer does not induce externality on any other consumer, resulting in a centrality of $a_F=1$. For a role model, on the other hand, from the definition of Katz-Bonacich centrality we can obtain
\begin{align}
a_L=1+(m-1)\delta(1-\beta)g+\delta(1-\beta)g a_L,
\end{align}
resulting in 
$a_L~=~\frac{1+(m-1)\delta(1-\beta)g}{1-\delta(1-\beta)g}$.
Similarly, we can find $b_F=1$ and $b_L~=~\frac{1+(m-1)\delta(1+\beta)g}{1-\delta(1+\beta)g}$.

Applying Theorem~\ref{theorem::Nash}, the symmetric Nash equilibrium strategy thus involves seeding a periphery consumer by
$\underline{s}^\star_F~=~\bar{s}^\star_F~=~p$, and every role model by an amount of
   \begin{align}
\label{eq::rolemodelL}
&\underline{s}^\star_L~=~\bar{s}^\star_L~=~\nonumber \\
&~~~\frac{p}{2}\left(\frac{1\!+\!(m\!-\!1)\delta(1\!-\!\beta)g}{1-\delta(1-\beta)g}+\frac{1\!+\!(m\!-\!1)\delta(1\!+\!\beta)g}{1-\delta(1+\beta)g}\right).
   \end{align}
Keeping the number of communities $\chi $ fixed and shifting $m\to\infty$, we can use Lemma~\ref{lemma:3} to verify that seeding only the $\chi$ role models according to \eqref{eq::rolemodelL} is an asymptotically sparse-realizable equilibrium strategy. Finally, Assumption~2 requires $\delta(1+\beta)\lambda_{\max}(G)<1$ which can be satisfied if $\delta(1+\beta)g<1$ noting that $\lambda_{\max}(G)=g$ (this follows from $d_i^{\text{in}}=g$ for all $i\in [n]$).
}

\section{Conclusion And Future Work}
\label{sec::conclusions}
We proposed and studied a strategic model of marketing and product consumption in social networks. Two firms offer substitutable products and compete to maximize the consumption of their products in a social network. Consumers are myopic and update their consumption level as the best response to the consumption of their neighbors in the previous period. This results in linear update dynamics for the product consumption. Moreover, each consumer receives externality from the consumption of each neighbor;  the externality is stronger for consumption of the same product. Firms can improve their market share by seeding the key consumers in the market, as their consumption will incentivize the consumption of the same product by their peers given the inter-agent externalities, which in turn can affect the consumption behavior all over the network.
We represented the above setting as a duopoly game between the firms and introduced a novel framework that allows for sparse seeding to asymptotically emerge as an equilibrium strategy.
We then studied the effect of the network structure on the optimal seeding strategies and the extent to which these strategies could be sparsified, under the proposed equilibrium concept. In particular, we derive necessary and sufficient conditions under which $\varepsilon$-Nash equilibrium strategies can asymptotically lead to sparse seeding in large populations. The results were demonstrated using a large core-periphery network structure with few highly connected and many sparsely connected nodes.
Extending our analysis to time-varying seeding strategies (recurring seeding), networks with uncertainties (e.g., in valuations of consumers about the products), and oligopoly setting are some of the potential venues for future research.

\bibliographystyle{IFAC}        
\begin{spacing}{1.4}
\bibliography{ref/main_Milad}
\end{spacing}


\end{document}